\documentclass[11pt]{article}

\usepackage{fullpage}

\usepackage{comment}
\usepackage{amsmath}
\usepackage{amsthm}
\usepackage{amssymb}
\usepackage{url}
\usepackage[final]{showkeys}
\usepackage[usenames,dvipsnames]{xcolor}
\usepackage{graphicx}

\usepackage{stmaryrd}

\usepackage[colorlinks=true, allcolors=blue]{hyperref}

\usepackage{latexsym}

\usepackage[vlined]{algorithm2e}
\DontPrintSemicolon
\SetVlineSkip{0pt}
\NoCaptionOfAlgo
\SetKwIF{If}{ElseIf}{Else}{if}{:}{else if}{else:}{}%
\SetKwFor{While}{while}{:}{}%
\SetKwProg{Fn}{}{:}{}

\newtheorem{theorem}{Theorem}[section]
\newtheorem{lemma}[theorem]{Lemma}

\newtheorem{corollary}[theorem]{Corollary}
\newtheorem{definition}[theorem]{Definition}

\newcommand{\OO}{\mathrm{O}}

\newcommand{\on}[2]{\mathop{\null#2}\limits^{#1}}
\newcommand{\bvec}[1]{\on{\,{}_\leftarrow}{{#1}}}

\newcommand{\RR}{\mathbb{R}}

\newcommand{\es}{\emptyset}
\newcommand{\sm}{\setminus}

\newcommand{\EQ}{\;=\;}

\newcommand{\GE}{\;\ge\;}

\newcommand{\Heap}{\mbox{{\it min\/}-{\it heap}}}
\newcommand{\Insert}{\mbox{{\it insert}}}
\newcommand{\DeleteMin}{\mbox{{\it delete\/}-{\it min}}}
\newcommand{\DecreaseKey}{\mbox{{\it decrease\/}-{\it key}}}

\newcommand{\Queue}{\mbox{{\it Queue}}}
\newcommand{\InsertLast}{\mbox{{\it insert\/}-{\it last}}}

\newcommand{\EBF}{\mbox{\sl e-BF\/}}
\newcommand{\EDJ}{\mbox{\sl e-Dijkstra\/}}

\title{Optimal energetic paths for electric cars}

\author{Dani Dorfman\thanks{Blavatnik School of Computer Science, Tel Aviv University, Israel. Email: {\tt dani.i.dorfman@gmail.com}, {\tt \{haimk,zwick\}@tau.ac.il.} Work of Uri Zwick partially supported by grant 2854/20 of the Israeli Science Foundation.} \and Haim Kaplan${}^*$ \and Robert E.\ Tarjan\thanks{Department of Computer Science, Princeton University. Email: {\tt ret@princeton.edu.} Partially supported by a gift from Microsoft.} \and Uri Zwick${}^*$}

\date{}

\begin{document}

\maketitle

\begin{abstract}%
\setlength{\parindent}{0pt}%
\setlength{\parskip}{3pt plus 2pt}%
\noindent
A weighted directed graph $G=(V,A,c)$, where $A\subseteq V\times V$ and $c:A\to \RR$, naturally describes a road network in which an electric car can roam. An arc $uv\in A$ models a road segment connecting the two vertices (junctions) $u$ and~$v$. The cost~$c(uv)$ of the arc $uv$ is the amount of \emph{energy} the electric car needs to travel from~$u$ to~$v$. This amount may be positive, zero or negative. To make the problem realistic, we assume that there are no negative cycles in the graph.

The electric car has a battery that can store up to~$B$ units of energy. The car can traverse an arc $uv\in A$ only if it is at~$u$ and the charge~$b$ in its battery satisfies $b\ge c(uv)$. If the car traverses the arc $uv$ then it reaches~$v$ with a charge of $\min\{b-c(uv),B\}$ in its battery. Arcs with a positive cost deplete the battery while arcs with negative costs charge the battery, but not above its capacity of~$B$. If the car is at a vertex~$u$ and cannot traverse any outgoing arcs of~$u$, then it is stuck and cannot continue traveling.

We consider the following natural problem: Given two vertices $s,t\in V$, can the car travel from~$s$ to~$t$, starting at~$s$ with an initial charge~$b$, where $0\le b\le B$? If so, what is the maximum charge with which the car can reach~$t$? Equivalently, what is the smallest $\delta_{B,b}(s,t)$ such that the car can reach~$t$ with a charge of $b-\delta_{B,b}(s,t)$ in its battery, and which path should the car follow to achieve this? We refer to $\delta_{B,b}(s,t)$ as the \emph{energetic cost} of traveling from~$s$ to~$t$. We let $\delta_{B,b}(s,t)=\infty$ if the car cannot travel from~$s$ to~$t$ starting with an initial charge of~$b$. The problem of computing energetic costs is a strict generalization of the standard shortest paths problem.

We show that the single-source version of the minimum energetic paths problem can be solved using simple, but subtle, adaptations of the classical Bellman-Ford and Dijkstra algorithms. To make Dijkstra's algorithm work in the presence of negative arc costs, but no negative cycles, we use a variant of the $A^*$ search heuristic.  This is similar to the idea used by Johnson to solve the standard all-pairs shortest paths problem.
These results are explicit or implicit in some previous papers. We provide a clearer, simpler and unified description of these algorithms.
\end{abstract}

\setcounter{page}{1}

\section{Introduction}\label{S:intro}

A weighted directed graph $G=(V,A,c)$, where $A\subseteq V\times V$ and $c:A\to \RR$, naturally describes a road network in which an electric car can roam. An arc $uv\in A$ models a road segment connecting the two vertices (junctions) $u$ and~$v$. The cost~$c(uv)$ of the arc $uv$ is the amount of \emph{energy} the electric car needs to travel from~$u$ to~$v$. This amount may be positive, e.g., if the road segment is uphill or level; zero; or negative, e.g., if the road segment is downhill. To make the problem realistic we assume that there are no negative cycles in the graph.

An electric car is equipped with a battery that can store up to~$B$ units of energy, where $B>0$ is a parameter. We assume that the electric car cannot be charged along the way and has to rely on the initial charge available in its battery. If the car is currently at vertex~$u$ with charge $b$ in its battery, where $0\le b\le B$, then it can traverse an arc $uv\in A$ if and only if $c(uv)\le b$. If this condition holds, and the car traverses the arc, then it reaches~$v$ with a charge of $\min\{b-c(uv),B\}$. In particular, the car can traverse $uv$ if $b-c(uv)>B$ (which can hold only if $c(uv)<0$), but the battery does not charge beyond its capacity of~$B$. We may assume that $c(uv)\in[-B,B]$ for every $uv\in A$, as arcs with $c(uv)>B$ can never be used, and thus can be removed, and costs $c(uv)<-B$ can be changed to $c(uv)=-B$.

We consider the following natural problem: Given two vertices $s,t\in V$, can the car travel from~$s$ to~$t$, starting at~$s$ with an initial charge~$b$, where $0\le b\le B$? If so, what is the maximum charge with which the car can reach~$t$? Equivalently, what is the smallest $\delta_{B,b}(s,t)$ such that the car can reach~$t$ with a charge of $b-\delta_{B,b}(s,t)$ in its battery, and which path should the car follow to achieve this? We refer to $\delta_{B,b}(s,t)$ as the \emph{energetic cost} of traveling from~$s$ to~$t$. We let $\delta_{B,b}(s,t)=\infty$ if the car cannot travel from~$s$ to~$t$ starting with an initial charge of~$b$. Note that the energetic cost depends on the capacity~$B$ of the battery and on the initial charge~$b$ of the battery at~$s$. Increasing the capacity~$B$ of the battery cannot increase energetic costs. Increasing the initial charge~$b$, on the other hand, can either decrease or increase the energetic cost $\delta_{B,b}(s,t)$. If the initial charge~$b$ is not large enough, the car will not be able to travel from~$s$ to~$t$, i.e., $\delta_{B,b}(s,t)=\infty$,
or the car may be forced to use energetically expensive detours.\footnote{For example, a two-arc path with costs $b$ and $-b$, for some $b>0$, modeling a mountain crossing, has an energetic cost of~$0$, but it can be used only if the initial charge of the battery is at least~$b$.} On the other hand, if the battery is almost fully charged, i.e., $b$ is close to~$B$, the car will not be able to take full advantage of downhill segments in the beginning the journey. In such cases increasing~$b$ may also increase $\delta_{B,b}(s,t)$.

We also consider the related problem of finding the \emph{minimum initial change} at~$s$ that will allow the car to travel to~$t$. We denote this quantity by $\beta_B(s,t)$. Note that $\beta_B(s,t)$ is the smallest~$b$ for which $\delta_{B,b}(s,t)<\infty$, or $\beta_B(s,t)=\infty$ if there is no such~$b$. We show that this problem is equivalent to computing energetic costs on the \emph{reverse} graph.

If all arc costs are non-negative, then $\delta_{B,b}(s,t)=\delta(s,t)$, if $\delta(s,t)\le b\le B$, where $\delta(s,t)$ is the standard distance from~$s$ to~$t$ in the graph~$G$. Otherwise, $\delta_{B,b}(s,t)=\infty$. More generally, even in the presence of negative arc costs, but no negative cycles, $\delta_{B,b}(s,t)=\delta(s,t)$ if and only if there exists a shortest path~$P$ from~$s$ to~$t$ such that the length of every prefix of~$P$ is in the interval $[b-B,b]$. In general, energetic costs may be larger than  distances, since the charge in the battery is constrained to remain in the interval $[0,B]$, i.e., it is not allowed to go negative and it is capped at~$B$. (For example, the electric car may not be able to traverse a mountain pass and may need to take a detour.) We do always have $\delta_{B,b}(s,t)\ge \delta(s,t)$.

The problem of computing minimum energetic costs is thus a strict generalization of the standard shortest paths problem. We show, however, that the single-source version of the energetic cost problem can still be solved using simple, but sometimes subtle, adaptations of the classical Bellman-Ford \cite{Bellman58,Ford56} and Dijkstra \cite{Dijkstra59} algorithms. To make Dijkstra's algorithm work in the presence of negative arc costs but no negative cycles, we use a variant of the $A^*$ search heuristic (see, e.g., Hart et al.~\cite{HNR68}).  This is similar to the idea used by Johnson~\cite{Johnson77} to solve the standard all-pairs shortest paths problem.

Unlike the standard shortest paths problem, the single-target version of the energetic costs problem is \emph{not} equivalent to the single-source version.  In particular, We cannot solve the single-target problem by running the adapted Bellman-Ford and Dijkstra algorithms backward.  Although there is always a tree of minimum energetic-cost paths from a source vertex~$s$ to all other vertices reachable from it, there are simple examples in which there is no tree of minimum energetic-cost paths to a target vertex~$t$, from all vertices that can reach it.

As mentioned in the abstract, the versions of the Bellman-Ford and Dijkstra algorithms presented here are explicit or implicit in some previous papers. Some other papers describe versions that are not as efficient as the versions described here. We try to present and prove the correctness of these two algorithms in the simplest and clearest possible way.

\subsection{Our results}\label{ss:our}

We show that the single-source version of the minimum energetic paths problem can be solved in $\OO(mn)$ time using a simple adaptation of the Bellman-Ford \cite{Bellman58,Ford56} algorithm, where $m=|A|$ and $n=|V|$. Furthermore, if a \emph{valid potential function} $p:V\to\RR$ is given, i.e., a function for which $c(uv)-p(u)+p(v)\ge 0$ for every $uv\in A$, then the single-source version can be solved in $\OO(m+n\log n)$ time using an adaptation of Dijkstra's \cite{Dijkstra59} algorithm similar to $A^*$ search (see, e.g., Hart et al.~\cite{HNR68}). Since a valid potential function can be found in $\OO(mn)$ time using the standard Bellman-Ford algorithm, the all-pairs version of the minimum energetic paths problem can be solved $\OO(mn+n^2\log n)$ time, essentially matching the bound for the standard APSP problem.

\subsection{Related results}\label{ss:related}

Various adaptations of the Bellman-Ford and Dijkstra algorithms for problems similar to or equivalent to the minimum energetic paths problem defined here were given by several authors. Artmeier et al.~\cite{artmeier2010shortest} give versions that run in $\OO(n^3)$ and $\OO(n^2)$ time, respectively. Brim and Chalupka \cite{BrCh12} give versions of these algorithms with the same running times as ours, but their descriptions of the algorithms are not easy to follow, as they define the problem differently and use some non-natural conventions, and the algorithms are only described as parts of a more complicated algorithm for the solution of one-player and two-player \emph{energy games}. (For more on energy games, see also Brim et al.~\cite{BCDGR11} and Dorfman et al.~\cite{DKZ19}.)  Baum et al.~\cite{baum2020energy} give a version of Dijkstra's algorithm but with a much slower running time.  They also show that the maximum charge with which~$t$ can be reached when starting at~$s$ with charge~$b$ is a piece-wise linear function of~$b$ with at most $\OO(n)$ breakpoints.

Khuller et al. \cite{khuller2011fill} consider a related problem in which the battery (or the fuel tank) can be recharged at intermediate vertices, with a possibly different price per unit of charge at each intermediate vertex. All arc costs are non-negative. They give various algorithms for computing a cheapest path from~$s$ to~$t$. Among these algorithms is a $\OO(n^2\Delta\log n)$-time algorithm for the single-target version, where~$\Delta$ is a bound on the number of rechargings allowed, and an $\OO(n^3\Delta^2)$-time algorithm for the all-pairs version.

Several authors, including Lehmann \cite{Lehmann77}, Tarjan \cite{Tarjan81a} and Mohri \cite{Mohri02} considered generalized versions of the shortest paths problem defined by \emph{semirings}. If $(R,\oplus,\otimes)$ is a semiring and $P$ is an $s$-$t$ path whose arcs have costs $c_i$, the cost of~$P$ is defined to be $c(P)=\otimes_{i=1}^k c_i$. The goal is to find $\oplus_P\, c(P)$, where $P$ ranges over all $s$-$t$ paths, assuming this quantity is well defined. The standard shortest paths problem corresponds to the \emph{tropical} semiring $(\RR,\min,+)$. All these results assume, as part of the definition of a semiring, that $\otimes$ is associative. Thus, as we shall see, none of these results apply to our problem, as our operation~$\otimes$ is not associative.

Generalized versions of Dijkstra's algorithm were obtained by various authors, most notably by Knuth \cite{Knuth77}. These generalization are of a different nature and are apparently not related to the version given here.

Other non-standard versions of the shortest paths problem were also considered. Perhaps the most famous one is the \emph{bottleneck} shortest paths problem. See, e.g., \cite{GaTa88,CKTZZ16} for the single-pair version, and \cite{VWY09,DuPe09a} for the all-pairs version. Vassilevska \cite{Vassilevska08} considered an interesting non-standard \emph{non-decreasing} version of the shortest paths problem related to reading train schedules.
Finally, Madani et al.~\cite{MTZ10} considered the \emph{discounted} shortest paths problem. All these problems are quite different from the problem considered here.

\subsection{Organization of the paper}\label{ss:organization}

In the next section we formally define the minimum energetic paths problem.
In Section~\ref{S:E-BF} we describe the modified version of the Bellman-Ford algorithm \cite{Bellman58,Ford56}.
In Section~\ref{S:E-Dijk} we describe the modified version of Dijkstra's algorithm~\cite{Dijkstra59}. In Section~\ref{S-min-init} we consider two closely related problem, minimum \emph{initial-energy} paths and maximum \emph{final-energy} paths. We end in Section~\ref{S:concl} with some concluding remarks and open problems.

\section{Minimum energetic paths}\label{S:minimum-charge}

To simplify the presentation, we concentrate on the computation of $\delta_B(s,t)=\delta_{B,B}(s,t)$, i.e., the energetic cost of traveling from~$s$ to~$t$ when starting with a fully charged battery of capacity~$B$. There is a simple reduction from the problem of computing $\delta_{B,b}(s,t)$, for an arbitrary $0\le b\le B$, to that of computing $\delta_B(s,t)$: Add a new vertex $s'$ and an arc $s's$ of cost $c(s's)=B-b$. Then, it is easy to see that $\delta_{B,b}(s,t)=\delta_B(s',t)-(B-b)$. A similar idea can be incorporated directly into the algorithms that we describe.

We begin with a definition of the energetic cost of a path.

\begin{definition}[Energetic cost of a path]
    If a path $P=u_0u_1\ldots u_k$ is traversable, when starting from~$u_0$ with a fully charged battery of capacity~$B$, then the final charge in the battery when reaching~$u_k$ is $B-d_{B}(P)$, where $d_{B}(P)$ is defined to be the
    the \emph{energetic cost} of the path. Note that $0\le d_{B}(P)\le B$. If the path is not traversable, we let $d_{B}(P)=\infty$.
\end{definition}
To obtain a simple formula for $d_{B}(P)$ we define the following operations:
\[ x\oplus_B y \EQ [x+y]_{\,0}^B \quad,\quad
 [z]_{\,0}^B \EQ \left\{\begin{array}{cl}
0 & \mbox{if $z<0$} \\
z & \mbox{if $0\le z\le B$} \\
\infty & \mbox{otherwise}
\end{array}\right. \]
We assume that $x+\infty=\infty+x=\infty$ for every $x\in[-B,B]\cup\{\infty\}$. For brevity, we sometimes write $x \oplus y$ instead of $x \oplus_B y$ when~$B$ is understood from the context.~\footnote{Note that $\oplus$ is not related to the semiring framework mentioned in Section~\ref{ss:related}.} Note that for every $x,y\in \RR$ and $B>0$ we have $x+y\le x\oplus_B y$. It is important to note that $\oplus_B$ is \emph{not} associative. (For example, $B\oplus(B\oplus -B)=B$ while $(B\oplus B)\oplus -B=\infty$, and $(-1\oplus -2)\oplus 2=2$ while $-1\oplus(-2\oplus 2)=0$, assuming $B\ge 2$.)

\begin{lemma}\label{L-d}
Let $P=u_0u_1\ldots u_k$ be a directed path, let $P'=u_0\ldots u_{k-1}$, and let $c_i=c(u_{i-1}u_{i})$, for $i=0,1,\ldots,k$. Let~$B$ be the capacity and initial charge of battery at~$u_0$. If $k=0$ then $d_{B}(P)=0$.
If $k>0$ then
\[ d_{B}(P) \EQ d_{B}(P')\oplus c_k \EQ
((\cdots((0 \oplus c_1)\oplus c_2)\oplus\cdots)\oplus c_{k-1})\oplus c_k \;.\]
\end{lemma}

\begin{proof}
Let $b_i$ be the charge of the battery at~$u_i$ and let $d_i=B-b_i$ be the \emph{depletion} of the battery at~$u_i$, for $i=0,1,\ldots,k$. Clearly $d_0=0$. It is easy to prove by induction that $d_i=d_{i-1}\oplus_B c_i$ and the lemma follows.
\end{proof}

As mentioned, the operation $\oplus_B$ is \emph{not} associative. Thus, in general, $d_{B}(P)\ne
0 \oplus (c_1 \oplus (\cdots \oplus(c_{k-2}\oplus(c_{k-1}\oplus c_k))\cdots))$.
In Section~\ref{S-min-init} we show, however that $c_1 \oplus (c_2 \oplus (\cdots \oplus(c_{k-1}\oplus(c_k\oplus 0))\cdots))$ also has an interesting meaning.

\begin{definition}[Energetic costs, minimum energetic paths]
    The \emph{energetic cost} $\delta_{B}(s,t)$ of traveling from~$s$ to~$t$ when starting from~$s$ with a fully charged battery of capacity~$B$ is defined as
\[ \delta_{B}(s,t) \EQ \min\{\, d_{B}(P) \mid \mbox{$P$ is an $s$-$t$ path in~$G$}\, \} \;. \]
If $\delta_{B}(s,t)<\infty$ and $P$ is an $s$-$t$ path satisfying $\delta_{B}(s,t)=d_{B}(P)$, then $P$ is said to be a \emph{minimum energetic path} from~$s$ to~$t$.
\end{definition}

Since we assume that there are no negative cycles in the graph, it is easy to see that for every path~$P$ from~$s$ to~$t$ there is a simple path $P'$ such that $d_{B}(P')\le d_{B}(P)$. (For this to hold it is in fact enough to require that there are no traversable negative cycles in the graph.) Thus, the minimum in the definition above can be taken over simple paths only.

It is interesting to note that $\delta_{B}(s,t)$ is well-defined also in the presence of traversable negative cycles, but minimum energetic paths are not necessarily simple in this case. In a companion paper \cite{DKTZ23} we obtain efficient algorithms that work in the presence of negative cycles. These algorithms, however, are substantially more complicated than the algorithms presented here and rely on new algorithmic techniques.

It is not difficult to see, and it will also follow from the correctness of the algorithms that we present in the next sections, that for every source vertex $s\in V$ there is always a tree of minimum energetic paths to all other vertices that can be reached from it. The simple example given in Figure~\ref{F-example} shows that a tree of minimum energetic paths to a given target vertex~$t$ does not always exist.

\begin{figure}[t]
\begin{center}
\includegraphics[scale=0.50]{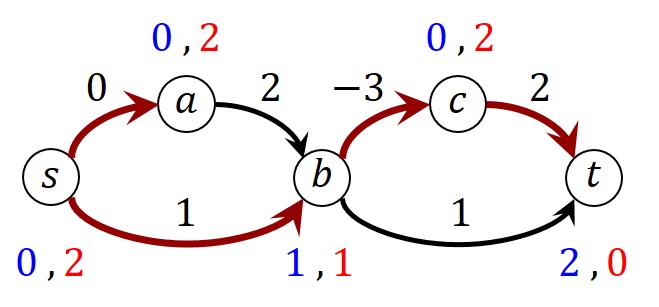}
\end{center}
\caption{A simple but illustrative example. Assume $B\ge 3$. The two numbers next to each vertex~$u$ are $\delta_B(s,u)$ and $\delta_B(u,t)$. The bold arcs constitute a tree of minimum energetic paths from the source vertex~$s$ to all other vertices. Another such tree can be obtained by replacing the arc $ct$ by $bt$. On the other hand, the only minimum path from $a$ to~$t$ is $abct$, while the only minimum path from~$b$ to~$t$ is the arc $bt$. Thus, there is no tree of minimum paths to~$t$ from all other vertices.}
\label{F-example}
\end{figure}

\section{An energetic version of the Bellman-Ford algorithm}\label{S:E-BF}

Recall that
for any $x$ and $y$, $x\oplus_B y \geq x+ y$.  This implies that in a graph without negative cycles,
if there is a traversable path from $s$ to $t$, there is such a path that is simple and hence contains at most $n-1$ arcs.  This means that if there are no negative cycles, we can solve the single-source minimum energetic paths problem using the Ford-Bellman \cite{Bellman58,Ford56} shortest path algorithm: We merely replace $+$ by $\oplus_B$.

We base our description of the algorithm on that in~\cite{Tarjan83}, which uses a queue as suggested by Gilsinn and Witzgall~\cite{GiWi73}.

The algorithm maintains a tentative energetic cost $d(v)$ for each vertex $v$, equal to the minimum of the energetic costs of paths from $s$ to $v$ found so far.   Initially $d(s)=0$ and $d(v)=\infty$ for $v \neq s$, where $s$ is the source.  It also maintains a queue $Q$, initially containing $s$.  The algorithm repeats the following step until $Q$ is empty:

\begin{quote}
Scan a vertex: Delete the front vertex $u$ on $Q$.  For each arc $uv$, if $\delta(u)\oplus_B c(uv) < \delta(v)$, \emph{relax} $uv$: Set $\delta(v) \gets \delta(u)\oplus_B c(uv)$, set $\pi(v)\gets u$, and add $v$ to the back of $Q$ if it is not on $Q$.
\end{quote}

Pseudocode of the algorithm, which we call $\EBF$, is given on the left of Figure~\ref{P-BF-D}. The correctness proof and analysis of the standard Bellman-Ford algorithm in the absence of negative cycles (see e.g., Tarjan~\cite{Tarjan83}) translates directly to this version.

\begin{theorem}\label{T-BF}
If $G=(V,A,c)$ has no traversable negative cycles
then $\EBF$ finds minimum energetic paths from~$s$ to all vertices in $\OO(mn)$ time.
\end{theorem}

\begin{proof}
   We define \emph{passes} over the queue.  Pass $0$ is the first scan step of~$s$.  Given that pass $k$ is defined, pass $k+1$ is the sequence of scan steps of vertices added to $Q$ during pass $k$.  A straightforward induction on $k$ shows that for each vertex $v$ that has a minimum-energy path of at most $k$ arcs, $d(v)$ is the energetic cost of such a path after $k$ passes.  It follows that the energetic costs are correctly computed. The $\pi$ values computed describe a tree of minimum-energy paths from~$s$ to all vertices reachable from~$s$ using a fully charged battery of capacity~$B$.
   Since each pass takes $\OO(m)$ time, the total time of the algorithm is $\OO(mn)$.
\end{proof}

In addition to the non-existence of negative cycles, the only thing required for correctness of the algorithm is that $\oplus_B$ is non-decreasing in its second argument: If $y \leq z$, $x\oplus y \leq x\oplus z$.

As in the standard version of the Bellman-Ford algorithm, one can add \emph{subtree disassembly}~\cite{Tarjan1981shortest,CGGTW09}, which does not improve the worst-case time bound but is likely to speed up the algorithm in practice. It is also easy to modify the algorithm so that it would find a traversable negative cycle that can be reached from~$s$, if any.

\begin{figure}[t]
\begin{center}
\begin{minipage}{2.8in}
\begin{algorithm}[H]
  \Fn{$\EBF(G=(V,A,c),B,s)$}{
  \BlankLine
    \For{$u\in V$}{
    $d(u)\gets \infty$ \;
    $\pi(u)\gets null$}
    \BlankLine
    $d(s)\gets 0$ \; 
    \BlankLine
    $Q\gets \Queue()$ \;
    $Q.\InsertLast(s)$ \;
    \BlankLine
    \While{$Q\ne\es$}{
    $u\gets Q.DeleteFirst()$ \;
    \For{$uv\in A$}{
    \If{$d(v)>d(u)\oplus_B c(uv)$}{
    $d(v)\gets d(u)\oplus_B c(uv)$ \;
    $\pi(v)\gets u$ \;
    \If{$v\notin Q$}
    {$Q.InsertLast(v)$ \;}
    }}
    }
    \BlankLine
    \Return $d$ \;
  }
\end{algorithm}
\end{minipage}
\begin{minipage}{3.5in}
\begin{algorithm}[H]
  \Fn{$\EDJ(G=(V,A,c),p,B,s)$}{
  \BlankLine
    \For{$u\in V$}{
    $d(u)\gets \infty$ \;
    $\pi(u)\gets null$ \;}
    \BlankLine
    $b(s)\gets 0$ \;
    \BlankLine
    $H\gets \Heap()$ \;
    $H.\Insert(s,p(s))$ \; 
    \BlankLine
    \While{$H\ne\es$}{
    $v\gets H.\DeleteMin()$ \;
    \For{$uv\in A$}
    {\If{$d(v)>d(u)\oplus_B c(uv)$}{
    $d(v)\gets d(u)\oplus_B c(uv)$ \;
    $\pi(v)\gets u$ \;
    \eIf{$v\notin H$}
    {$H.\Insert(u,d(v)+p(v))$ \;}
    {$H.\DecreaseKey(u,d(v)+p(v))$ \;}
    }}
    }
    \BlankLine
    \Return $d$ \;
  }
\end{algorithm}\end{minipage}
\end{center}
\caption{Energetic variants of the Bellman-Ford and Dijkstra algorithms.}\label{P-BF-D}
\end{figure}

The correctness of \EBF\ implies the following corollary, which we use to the prove the correctness of the energetic variant of Dijkstra's algorithm:

\begin{corollary}\label{C-verify}
    If each $d(v)<\infty$ corresponds to the energetic cost of some path from~$s$ to~$v$, and $d(v)\le d(u)\oplus_B c(uv)$ 
    for every $uv\in A$, then $d(v)=\delta_B(s,v)$, for every $v\in V$.
\end{corollary}

\section{An energetic version of Dijkstra's algorithm}\label{S:E-Dijk}

If all arc costs are non-negative, Dijkstra's algorithm \cite{Dijkstra59}  with $\oplus_B$ replacing $+$ will solve the one-source problem.  This algorithm replaces the queue $Q$ in the Bellman-Ford algorithm with a heap~$H$.  The key of a vertex $v$ in the heap is $d(v)$.  Each scan step deletes a vertex of minimum key from the heap.  When a relaxation decreases the key of a vertex in the heap, the algorithm does the appropriate \emph{decrease-key} operation on the heap.  If all arc costs are non-negative, the algorithm deletes each vertex from~$H$ at most once, and when a vertex $v$ is deleted from~$H$, $d(v)$ is the minimum energetic cost of a path from $s$ to $v$.  The proof of correctness mimics that of the standard Dijsktra algorithm. The algorithm does at most $n$ heap insertions, at most $n$ heap deletions, and at most $m$ decrease-key operations.  If the heap is a Fibonacci heap~\cite{FrTa87} or equally efficient data structure, e.g., \cite{HKTZ17}, the total running time is $\OO(m+n\log n)$. In fact, the algorithm is identical to the standard algorithm with~$d(v)$ values greater than~$B$ replaced by~$\infty$.

More interesting is that if arc costs can be negative but there are no negative cycles, we can use a variant of the $A^*$ search algorithm, which is a modification of Dijkstra's algorithm, to solve the one-source minimum energetic paths problem in $\OO(m+n\log n)$ time, provided that we have a \emph{valid potential function} $p:V\to\RR$. A potential $p$ is \emph{valid} if $c(uv)-p(u)+p(v) \ge 0$ for every arc $uv\in A$. It is well-known that a valid potential function exists if and only if the graph contains no negative cycles.

The $A^*$ search algorithm is almost identical to Dijkstra's algorithm. The only difference is that the \emph{key} of vertex~$v$ in the heap is $d(v)+p(v)$, and not just $d(v)$, where $p$ is a valid potential function. In the original setting of the $A^*$ search heuristic, $p(v)$ is an estimate of the distance from~$v$ to the destination~$t$. The correctness of the algorithm only requires, however, that $p$ is a valid potential function. If $p$ is valid, the $A^*$ algorithm deletes each vertex $v$ at most once from the heap, and when~$v$ is deleted, $d(v)=\delta_B(s,v)$, the energetic cost of traveling from~$s$ to~$v$.

An energetic version of the $A^*$ is obtained simply by replacing $+$ by $\oplus_B$ in relaxations. We assume the algorithm is given a potential $p$ that is valid for $+$, not $\oplus_B$. The algorithm begins with $d(s)=0$ and $d(v)=\infty$ for each vertex $v\in V\sm\{s\}$, and $H$ containing $s$.  The key of a vertex $v$ in $H$ is $d(v)+p(v)$. The algorithm repeats the following step until $H$ is empty:

\begin{quote}
  Scan a vertex: Delete from $H$ a vertex $u$ with minimum key $d(u)+p(u)$. For each arc~$uv$, if $d(u)\oplus_B c(uv) < d(v)$, \emph{relax} $uv$: Set $d(v) \gets d(u)\oplus_B c(uv)$; $\pi(v)\gets u$; add $v$ to $H$ with key $d(v)+p(v)$ if $v\notin H$, or decrease the key of $v$ to $d(v)+p(v)$ if $v\in H$.
\end{quote}

Pseudocode of the resulting algorithm, which we call \EDJ, is given on the right of Figure~\ref{P-BF-D}. The main step towards establishing the correctness of \EDJ\ is the following:

\begin{lemma}\label{L:A*}
If $p$ is a valid potential then \EDJ\ maintains the following invariant: if $u$ has been deleted from $H$ while $v$ has not been deleted from~$H$ yet, then $d(u)+p(u)\leq d(v)+p(v)$. As a consequence, each vertex~$u$ is inserted and deleted from~$H$ at most once.
\end{lemma}
\begin{proof}
We prove the lemma by induction on the number of heap operations.  The lemma is true initially, as no vertex was deleted from~$H$ yet.  Suppose it is true just before $u$ is deleted from $H$.  Since $d(u)+p(u)$ is \emph{minimum} among all $u\in H$, and since $d(u)=\infty$ for all vertices not yet inserted into~$H$, the invariant holds just after $u$ is deleted from $H$. By the induction hypothesis, $d(u)+p(u)$ is now \emph{maximum} over all $u$ already deleted from~$H$. Suppose the invariant holds just before the relaxation of an arc $uv$.  Just after the relaxation, $d(v)=d(u)\oplus_B c(uv) \geq d(u) + c(uv)$.  Hence
\[d(v)+p(v) \GE d(u)+c(uv)+p(v) \GE d(u) + p(u)\;,\] where the last inequality follows by the validity of $p$. Since the relaxation strictly decreased $d(v)$, it follows that $v$ could not have already been deleted from~$H$, as it would violate the claim that $d(u)+p(u)$ is maximum over all vertices already deleted from~$H$. Thus, $v$ is either in~$H$ or was not inserted into~$H$ yet. Decreasing the key of~$v$ to $d(v)+p(v)$, or inserting $v$ to $H$ with this key does not violate the invariant.
\end{proof}

The proof of Lemma~\ref{L:A*} is the same as the proof of the corresponding lemma for the standard version of $A^*$ except for the use of the inequality $x \oplus_B y \geq x+y$. Using Lemma~\ref{L:A*} we can easily prove the correctness of the algorithm.

\begin{theorem}\label{T-EDJ}
If $G=(V,A,c)$ has no negative cycles and $p$ is a valid potential for~$G$,
then $\EDJ$ finds minimum energetic paths from~$s$ to all vertices in $\OO(m+n\log n)$ time.
\end{theorem}

\begin{proof}
    When a vertex $u$ is removed from~$H$, all outgoing arcs $uv$ are scanned and all appropriate relax operations are performed. By Lemma~\ref{L:A*}, $d(u)$ will not be changed again. Thus, when the algorithm terminates $d(v)\le d(u)\oplus_B c(uv)$ for every arc $uv\in A$. By Corollary~\ref{C-verify}, we have $d(v)=\delta_B(s,v)$, for every $v\in V$. As in the proof of Theorem~\ref{T-BF} we get that the $\pi$ values describe a tree of minimum energetic paths from~$s$ to all vertices that can be reached from~$s$.

    The algorithm performs at most $n$ heap insertions, at most $n$ heap deletions, and at most $m$ decrease-key operations. With an efficient heap implementation the total running time is $\OO(m+n\log n)$.
\end{proof}

The remaining question is how to obtain a valid potential.  For this we can use any standard shortest path algorithm: If $s$ is an arbitrary source from which all vertices are reachable, there are no negative cycles, and $p(v)=-\delta(s,v)$, where $\delta(s,v)$ is the standard distance from~$s$ to~$v$, then $c(uv)-p(u)+p(v)\ge 0$, for every arc $uv$, by the triangle inequality. (If there is no such vertex~$s$ in the graph, add a new vertex~$s$ and connect it with zero-cost arcs to all other vertices.)

Thus we can compute minimum energetic paths from $k$ sources in $\OO(m+n\log n)$ time per source plus the time to solve one standard one-source shortest path problem with the given arc costs.  The extra time needed for this preprocessing is $\OO(mn)$ if we use Bellman-Ford, $\OO(mn^{1/2}\log C)$ if we use Goldberg's \cite{Goldberg95} shortest path algorithm, or $\OO(m \log^2 n \log(nC)\log\log n)$ time if we use the improvement by Bringmann et al.~\cite{BCF23} of the recent breakthrough result of Bernstein et al.~\cite{BernsteinNW22}; the latter two bounds require integer arc costs no smaller than $-C$. With any one of these choices, we obtain the following corollary:

\begin{corollary}
    The all-pairs minimum energetic paths in a graph $G=(V,A,c)$ with no negative cycles can be solved in $\OO(mn+n^2\log n)$ time.
\end{corollary}

The resulting all-pairs algorithm is very similar to Johnson's \cite{Johnson77} algorithm for the standard all-pairs shortest paths problem. The slight advantage of using the $A^*$ formulation, as we have done, is that no explicit transformation of arc costs is needed, only a simple modification of heap keys. A potential transformation of arc costs was also used by Edmonds and Karp \cite{EdKa72}.

\section{Minimum initial charges and maximum final charges}\label{S-min-init}

To end the paper, we consider two problems that are closely related to the minimum energetic paths problem. Let $G=(V,A,c)$ be a graph with no (traversable) negative cycles and let~$B$ be the capacity of the battery. For two vertices $s,t\in V$, we let $\alpha_B(s,t)$ be the \emph{maximum final charge} with which it is possible to reach~$t$ when starting at~$s$ with a full battery, or $-\infty$, if it is not possible to travel from~$s$ to~$t$. We also let $\beta_B(s,t)$ be the \emph{minimum initial charge} required at~$s$ for getting to~$t$, or $\infty$, if no initial charge (of at most~$B$) is sufficient.

The maximum final charge problem is not really a new problem as $\alpha_B(s,t)=B-\delta_B(s,t)$. As noted in the introduction, $\beta_B(s,t)$ is the smallest $b$ such that $\delta_{B,b}(s,t)\le b$, or equivalently $\delta_{B,b}(s,t)<\infty$, if there is such a~$b$. Thus, if $B$ and all arc costs are integral, then $\beta_B(s,t)$, for a specific pair~$s$ and~$t$, can solved using binary search.

There is, however, a more interesting relation between the minimum initial-charge problem and the minimum energetic cost problem.
Namely, $\beta_B(s,t)$ is equal to $\delta^{\bvec{G}}_B(t,s)$ the energetic cost of traveling from~$t$ to~$s$ in the \emph{reversed graph} $\bvec{G}$, the graph obtained by reversing all the arcs in the graph~$G$ and retaining all arc costs. This relation follows easily from the following lemma, analogous to Lemma~\ref{L-d}, whose simple proof is omitted. For a path~$P$ from~$s$ to~$t$, let $b_B(P)$ be the minimum initial charge at~$s$ with which the path~$P$ can be traversed.

\begin{lemma}\label{L-b}
Let $P=u_0u_1\ldots u_k$ be a directed path, let $P'=u_1\ldots u_{k}$, and let $c_i=c(u_{i-1}u_{i})$, for $i=1,\ldots,k$. Let~$B$ be the capacity of the battery. If $k=0$ then $b_{B}(P)=0$.
If $k>0$ then
\[ b_{B}(P) \EQ c_1 
\oplus_B b_{B}(P') \EQ
c_1 \oplus (c_2 \oplus (\cdots \oplus(c_{k-1}\oplus(c_k\oplus 0))\cdots)) \;.\]
\end{lemma}

\begin{corollary}
    For every $s,t\in V$, $\beta_B(s,t)=\delta^{\bvec{G}}_B(t,s)$.
\end{corollary}

As immediate corollaries, it follows that we can solve the single-target version of the minimum initial-charge paths problem in $\OO(m+n\log n)$ time, 
if we are given a valid potential, and the all-pairs version of the problem in $\OO(mn+n^2\log n)$.

\section{Concluding remarks and open problems}\label{S:concl}

Our goal was to present simple algorithm for solving the minimum energetic paths problem in the absence of negative cycles. Most of the results in this paper are not new. Our formulation of the minimum energetic paths problem, which we believe helps simplify and clarify the results and their relation to standard shortest paths algorithms, seems to be new, however. The simple relation between minimum energetic paths and minimum initial-charge paths also seems to be new.

An interesting feature of the minimum energetic paths problem is that it is not symmetric, i.e., the single-target version of the problem is not equivalent to the single-source version. The best algorithm that we currently have for the single-target version actually solves the all-pairs version of the problem. It would be interesting to know whether there is a more efficient solution. The fact that there may not be a tree of minimum energetic paths to a given target may indicate that the single-target version is harder than the single-source version.

In a companion paper \cite{DKTZ23} we extend some of the results presented here to the setting in which negative cycles may be present. The algorithms required in that case are much more complicated.

\bibliographystyle{plain}
\bibliography{bibliography}

\end{document}